\documentclass[twocolumn,english,showpacs,nofootinbib,nobibnotes,aps,pra]{revtex4-1}

\usepackage[latin9]{inputenc}
\usepackage{units}
\usepackage{amsthm}
\usepackage{amsmath}
\usepackage{amssymb}
\usepackage{graphicx}
\usepackage{dcolumn}   
\usepackage{bm}        
\usepackage{multirow}
\usepackage{amssymb}
\usepackage{amsfonts}  

\pdfpageheight\paperheight
\pdfpagewidth\paperwidth

\theoremstyle{plain}
\newtheorem{thm}{\protect\theoremname}
\theoremstyle{definition}
\newtheorem{defn}[thm]{\protect\definitionname}
\newtheorem{corollary}[thm]{\protect Corollary}
\ifx\proof\undefined
\newenvironment{proof}[1][\protect\proofname]{\par
\normalfont\topsep6\p@\@plus6\p@\relax
\trivlist
\itemindent\parindent
\item[\hskip\labelsep\scshape #1]\ignorespaces
}{%
\endtrivlist\@endpefalse
}
\providecommand{\proofname}{Proof}
\fi

\usepackage{braket}
\usepackage{dsfont}

\makeatother

\usepackage{babel}
\providecommand{\definitionname}{Definition}
\providecommand{\theoremname}{Theorem}

\global\long\def\trace{\operatorname{Tr}}
\global\long\def\ketbra#1#2{\ket{#1}\!\bra{#2}}
\global\long\def\real{\operatorname{Re}}
\global\long\def\one{\mathds{1}}

\begin{document}

\title{Almost all four-particle pure states are determined \\
by their two-body marginals}

\author{Nikolai Wyderka}
\affiliation{Naturwissenschaftlich-Technische Fakultät, Universität Siegen, 
Walter-Flex-Str.~3, D-57068 Siegen, Germany}

\author{Felix Huber}
\affiliation{Naturwissenschaftlich-Technische Fakultät, Universität Siegen, 
Walter-Flex-Str.~3, D-57068 Siegen, Germany}

\author{Otfried Gühne}
\affiliation{Naturwissenschaftlich-Technische Fakultät, Universität Siegen, 
Walter-Flex-Str.~3, D-57068 Siegen, Germany}

\date{\today}

\begin{abstract}
We show that generic pure states (states drawn according to the Haar measure) of four
particles of equal internal dimension are uniquely determined among all other pure states
by their two-body marginals. In fact, certain subsets of three of the two-body marginals
suffice for the characterization. We also discuss generalizations of the statement to pure
states of more particles, showing that these are almost always determined among pure states
by three of their $(n-2)$-body marginals. Finally, we present special families of symmetric
pure four-particle states that share the same two-body marginals and are therefore
undetermined. These are four-qubit Dicke states in superposition with generalized GHZ states.
\end{abstract}
\maketitle

{\it Introduction.}
The question of what can be learned about a multiparticle system by looking
at some particles only is central for many problems in physics. In quantum
mechanics, this problem can be formulated in a mathematical fashion as 
follows: Given a quantum state $\rho$ on $n$ particles, which properties
of this state can be inferred from knowledge of the $k$-particle reduced states
only? This question is naturally connected to the phenomenon of entanglement.
Indeed, considering pure states of two-particles, product states are always 
determined by their marginals, whereas entangled states can exhibit reduced 
states that admit multiple compatible joint states. Consequently, entangled
states may contain information in correlations among many parties that is lost 
when just having access to the reductions. In fact, many works have considered
the problem how entanglement or other global properties relate to properties of 
the reduced states \cite{tkgb09, wuerflinger12, Walter1205, miklin16}. On a more fundamental level, 
one may ask the question whether for a given set of reduced states the original 
global state is the only state having this set of reduced states \cite{coleman63, klyachko04, sawicki13, schilling17}.

This question is also of practical interest: If a quantum state happens to be 
the unique ground state of a Hamiltonian, it may be obtained by engineering this
Hamiltonian and then cooling down the system. In practice, typical Hamiltonians 
are limited to having interactions between two or three particles only. The question 
of whether the ground state of such a Hamiltonian is unique is then directly related 
to the question of whether the state one wants to prepare is uniquely determined 
by its two- or three-body marginals \citep{zhou09, huber2016characterizing}.

The question of uniqueness was analyzed in detail by Linden and coworkers, who showed that 
almost all pure three-qubit states are determined among all mixed states by their 
two-body marginals \citep{linden2002almost}. Later, Diósi showed that two of the three 
two-body marginals suffice to characterize uniquely a pure three-particle state 
among all other pure states \citep{diosi2004three}. Jones and Linden finally proved 
that generic states of $n$ qudits are uniquely determined by certain sets of reduced 
states of just more than half of the parties, whereas the reduced states of fewer 
than half of the parties are not sufficient \citep{jones2005parts}. Thus, higher-order 
correlations of most pure quantum states are not independent of the lower-order 
correlations.

In this paper, we investigate the case of four-particle states having equal internal 
dimension. We show that generic pure states of four particles are uniquely determined 
among all pure states by certain sets of their two-body marginals. To that end, we 
begin by defining what we mean by generic states and distinguish the different kinds 
of uniqueness, namely uniqueness among pure and uniqueness among all states. We then 
prove our main result, first for the case of qubits and 
subsequently the general case of qudits. The theorem is then generalized to generic
$n$-particle states, which can be shown to be determined in a similar way by certain
sets of three of their $(n-2)$-body marginals.
Finally, we list some specific examples for the exceptional case of states of four
particles that are not determined by their two-body marginals.

\begin{figure}[t!]
\begin{centering}
\includegraphics[width=0.85\columnwidth]{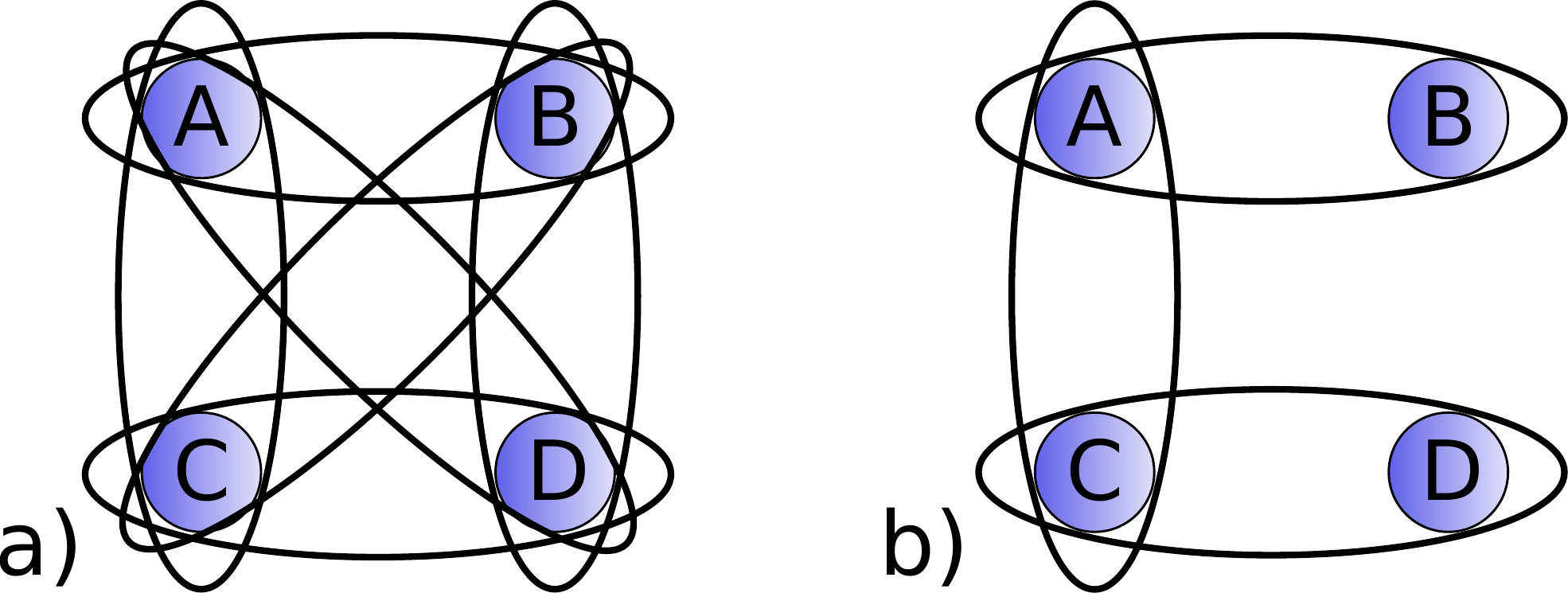}
\par\end{centering}
\protect\caption{Illustration of two different sets of two-body marginals: 
a) the set of all six two-body marginals, b) a set of three two-body marginals
that is shown to suffice to uniquely determine pure generic pure states.
\label{fig:Illustration-of-the}}
\end{figure}

{\it Random states and uniqueness.}
We begin with some basic definitions. Given an $n$-particle quantum
state $\rho$ of parties $\mathcal{P}=\{ P_{1},\ldots,P_{n}\} $,
its $k$-body marginal of parties $\mathcal{S}=\{ P_{i_{1}},\ldots,P_{i_{k}}\}$
is defined as
\begin{equation}
\rho_{\mathcal{S}}:=\trace_{\bar{\mathcal{S}}}(\rho)\:,
\end{equation}
where the trace is a partial trace over parties $\bar{\mathcal{S}}=\mathcal{P}\setminus \mathcal{S}$.
When stating the question of uniqueness, i.e., whether a given state
is uniquely determined by some of its marginals, it is important to specify
the set of states considered. Usually, two different sets are taken into
account,
namely the set of pure states and the set of all states, leading to
two different kinds of uniqueness, namely that of uniqueness among
pure states (UDP) and uniqueness among all states (UDA). We adopt
here the definition of Ref.~\citep{chen2013uniqueness} and extend
it by specifying which marginals are involved:
\begin{defn}
A state $\ket{\psi}$ is called
\begin{itemize}
\item $k$-\emph{uniquely determined among pure states} ($k$-UDP), if there
exists no other pure state having the same $k$-body marginals as
$\ket{\psi}$.
\item $k$-\emph{uniquely determined among all states} ($k$-UDA), if there
exists no other state (mixed or pure) having the same $k$-body marginals
as $\ket{\psi}$.
\end{itemize}
\end{defn}
Using this language, the results of Ref.~\citep{linden2002almost}
show that almost all three-qubit pure states are $2$-UDA, that is,
given a random pure state $\ket{\psi}$, it is uniquely determined
by its marginals $\rho_{\text{AB}}$, $\rho_{\text{AC}}$ and $\rho_{\text{BC}}$.
Ref.~\citep{diosi2004three} states that knowledge of just two of
the three two-body marginals suffices to fix the state among
all pure states (UDP). Later, these results were generalized to states of
certain higher internal dimensions, for a more general overview see
for example Ref.~\citep{chen2013uniqueness}. Note that while UDA
implies UDP, the converse in general does not need to be true and 
there are examples of four-qubit states which are 2-UDP but not 2-UDA
\citep{xin2017quantum}. Other cases of UDP versus UDA are discussed
in Ref.~\citep{chen2013uniqueness}.

Note that in some cases a subset of all $k$-body marginals already suffices
to show uniqueness, as in the case of almost all three-qubit states
discussed above \citep{diosi2004three}. In this paper, we will show
that in case of four particles, specific sets consisting of three
of the six two-body marginals suffice to determine any generic pure
states among all pure states.

Generic states are understood to be states drawn randomly according
to the Haar measure. Here, we adopt a special procedure to obtain
such random states in a Schmidt decomposed form. To that end, 
consider a four-particle pure state 
$\ket{\psi}\in\mathcal{H}_{A}\otimes\mathcal{H}_{B}\otimes\mathcal{H}_{C}\otimes\mathcal{H}_{D}$,
where $\dim\mathcal{H}_{A}=\dim\mathcal{H}_{B}=\ldots=d$. Using the
Schmidt decomposition along the bipartition ($AB \vert CD$), 
the state can be written as
\begin{equation}
\ket{\psi}=\sum_{i=1}^{d^{2}}\sqrt{\lambda_{i}}\ket{i}_{\text{AB}}\otimes\ket{i}_{\text{CD}},\label{eq:schmidtdeco}
\end{equation}
where $\sum_{i}\lambda_{i}=1$. If the state has full Schmidt rank,
i.e.,~$\lambda_{i}\neq0$ for all $i$, then the sets $\ket{i}_{\text{AB}}$
and $\ket{i}_{\text{CD}}$ form orthonormal bases in the composite
Hilbert spaces $\mathcal{H}_{A}\otimes\mathcal{H}_{B}$ and 
$\mathcal{H}_{C}\otimes\mathcal{H}_{D}$, respectively.

\begin{defn}
A \emph{generic} four-particle pure state is a state 
$\ket{\psi}\in\mathcal{H}_{A}\otimes\mathcal{H}_{B}\otimes\mathcal{H}_{C}\otimes\mathcal{H}_{D}$
drawn randomly according to the Haar measure. Writing such state 
as in Eq.~(\ref{eq:schmidtdeco}), the Schmidt bases and the set of 
Schmidt coefficients are independent from each other. 
The distribution of the Schmidt coefficients
is given by \citep{scott2003entangling,lloyd1988complexity}
\begin{multline}
P(\lambda_{1},\ldots,\lambda_{d^2})\text{d}\lambda_{1}\ldots\text{d}\lambda_{d^2}\\
=N\delta\left(1-\sum_{i=1}^{d^2}\lambda_{i}\right)\prod_{1\leq i<j\leq d^2}
(\lambda_{i}-\lambda_{j})^{2}\text{d}\lambda_{1}\ldots\text{d}\lambda_{d^2}
\end{multline}
and the Schmidt bases are distributed according to the Haar measure
of unitary operators on the smaller Hilbert spaces.
\end{defn}

The mutual independence of the two Schmidt bases and the coefficients
can be seen from the fact that in the Haar measure, for the probability
distribution $p(\ket{\psi})$ to obtain state $\ket{\psi}$ holds
$p(\ket{\psi})=p(\one_{\text{AB}}\otimes U_{\text{CD}}\ket{\psi}) = 
p(U_{\text{AB}}\otimes\one_{\text{CD}}\ket{\psi})$. 

Generic states as defined 
above exhibit two other important properties: They have full Schmidt rank 
and pairwise distinct Schmidt coefficients. We would like to add that while 
the definition above makes use of the Haar measure, we do not explicitly 
require it. Any measure with the same independence properties between 
the two Schmidt
bases and Schmidt coefficients would work as well, as long as the sets of 
states having non-full Schmidt rank or degenerate Schmidt coefficients are 
also of measure zero.

{\it The case of qubits.}
To begin with, we investigate the qubit case, where $d=2$. Let 
$\ket\psi=\sum_{i=1}^{4}\sqrt{\lambda_{i}}\ket{i}_{\text{AB}}\otimes\ket{i}_{\text{CD}}$
be a generic state in the sense defined above. The two-body marginal of parties $A$ and $B$ 
is given by
\begin{equation}
\rho_{\text{AB}}=\trace_{\text{CD}}(\ketbra{\psi}{\psi})=\sum_{i=1}^{4}\lambda_{i}\ketbra ii_{\text{AB}}\label{eq:rhoab}
\end{equation}
and similarly for $CD$. This is the starting point for the proof
of the following theorem.

\begin{thm}
Almost all four-qubit pure states are uniquely determined among pure
states by the three two-body marginals $\rho_{\text{AB}}$, $\rho_{\text{CD}}$
and $\rho_{\text{BD}}$. In particular, this implies that they are
$2$-UDP.
\end{thm}

\begin{proof}
Let $\ket\psi$ be a generic state in the Schmidt decomposed form
in Eq.~(\ref{eq:schmidtdeco}). We arrange the Schmidt bases such
that the Schmidt coefficients are in decreasing order, i.e.~$\lambda_{i}\geq\lambda_{i+1}$.
Suppose that there is another pure state $\ket\phi$ which exhibits
the same two-body marginals $\rho_{\text{AB}}$ and $\rho_{\text{CD}}$
as $\ket\psi$. As the $\lambda_{i}$ are pairwise distinct and in
decreasing order, the Schmidt bases of $\ket\phi$ and $\ket\psi$
have to coincide up to a phase. Thus, $\ket\phi$ must be of the form
\begin{equation}
\ket\phi=\sum_{i=1}^{4}e^{i\varphi_{i}}\sqrt{\lambda_{i}}\ket{i}_{\text{AB}}\otimes\ket{i}_{\text{CD}}.
\end{equation}
Therefore, the only degrees of freedom left of $\ket\phi$ are the
four phases $\varphi_{i}$.

We now demand that also the marginals of parties $B$ and $D$ coincide,
i.e.~$\trace_{\text{AC}}(\ketbra{\psi}{\psi})=\trace_{\text{AC}}(\ketbra{\phi}{\phi})$
(but any other marginal would be fine, too): 
\begin{eqnarray}
\!\!\!\rho_{\text{BD}} & \!=\!\! & \sum_{i,j=1}^{4}\!\!\sqrt{\lambda_{i}\lambda_{j}}\trace_{\text{AC}}(\ketbra ij_{\text{AB}}\otimes\ketbra ij_{\text{CD}})\nonumber \\
\!\!\! & \!\overset{!}{=}\!\! & \sum_{i,j=1}^{4}\!\! e^{i(\varphi_{i}-\varphi_{j})}\sqrt{\lambda_{i}\lambda_{j}}\trace_{\text{AC}}(\ketbra ij_{\text{AB}}\otimes\ketbra ij_{\text{CD}})\text{.}\label{eq:lindepeq}
\end{eqnarray}
The sum runs over operators on the space of parties $B$ and $D$.
For every pair $i,j$, this operator is given by 
\begin{equation}
O_{ij}=\trace_{\text{AC}}(\ketbra ij_{\text{AB}}\otimes\ketbra ij_{\text{CD}}).
\end{equation}
The 16 operators $O_{ij}$ span a subspace in the 16-dimensional space
of operators on $\mathcal{H}_{B}\otimes\mathcal{H}_{D}$. As we will
see later, this subspace is only 13-dimensional, thus the operators
must be linearly dependent. Therefore, we cannot simply compare both
sides of Eq.~(\ref{eq:lindepeq}) term by term to conclude that $\varphi_{i}=\varphi_{j}$.
Instead, let us interpret the 16 operators $O_{ij}$ as vectors in
the 16-dimensional operator space. Thus, we are looking for solutions
of the equation
\begin{equation}
\sum_{i,j=1}^{4}(1-e^{i(\varphi_{i}-\varphi_{j})})\sqrt{\lambda_{i}\lambda_{j}}O_{ij}\equiv\sum_{i,j=1}^{4}\gamma_{ij}O_{ij}=0_{4\times4}\:,\label{eq:gammaO}
\end{equation}
where 
\begin{equation}
\gamma_{ij}:=(1-e^{i(\varphi_{i}-\varphi_{j})})\sqrt{\lambda_{i}\lambda_{j}}\:\text{.}\label{eq:gamma}
\end{equation}
These are 16 equations, one for every entry of the resulting $4\times4$
matrix. We can treat Eq.~(\ref{eq:gammaO}) as a system of linear
equations for the $\gamma_{ij}$ and look for solutions that can be
written in the specific form in Eq.~(\ref{eq:gamma}). It implies
that
\begin{eqnarray}
\gamma_{ii} & = & 0\:,\label{eq:gammaprop1}\\
\gamma_{ij} & = & \bar{\gamma}_{ji}\:,\label{eq:gammaprop2}
\end{eqnarray}
Therefore, there are effectively six undetermined complex-valued variables
$\gamma_{ij}$ for $1\leq i<j\leq4$.

Let us now investigate the linear system in Eq.~(\ref{eq:gammaO})
in more detail. Note that every $O_{ij}$ can be written as a product
\begin{equation}
O_{ij}=\trace_{\text{A}}(\ketbra ij_{\text{AB}})\otimes\trace_{\text{C}}(\ketbra ij_{\text{CD}})\equiv Q_{ij}\otimes R_{ij}\:,
\end{equation}
where $Q_{ij}=\trace_{\text{A}}(\ketbra ij_{\text{AB}})$, $R_{ij}=\trace_{\text{C}}(\ketbra ij_{\text{CD}})$.
The matrices $Q_{ij}$ and $R_{ij}$ inherit some properties from
the underlying orthonormal bases:
\begin{eqnarray}
\trace(Q_{ij}) & = & \delta_{ij}\:,\nonumber \\
Q_{ij}^{\dagger} & = & Q_{ji}\label{eq:properties}
\end{eqnarray}
and similarly for $R_{ij}$.

Using these properties together with Eqs.~(\ref{eq:gammaprop1})
and (\ref{eq:gammaprop2}), Eq.~(\ref{eq:gammaO}) can be written
as
\begin{equation}
\sum_{i<j}\gamma_{ij}Q_{ij}\otimes R_{ij}+\bar{\gamma}_{ij}Q_{ij}^{\dagger}\otimes R_{ij}^{\dagger}\overset{!}{=}0\:.
\end{equation}

For $i\neq j$, $\trace(Q_{ij})=\trace(R_{ij})=0$ and we can write
$Q_{ij}$ and $R_{ij}$ explicitly as 
\begin{eqnarray}
Q_{ij} & = & \begin{pmatrix}q_{ij}^{11} & q_{ij}^{12}\\
q_{ij}^{21} & -q_{ij}^{11}
\end{pmatrix},\\
R_{ij} & = & \begin{pmatrix}r_{ij}^{11} & r_{ij}^{12}\\
r_{ij}^{21} & -r_{ij}^{11}
\end{pmatrix}.
\end{eqnarray}
Thus, 
\begin{eqnarray}
0 & \!=\!\!\! & \sum_{i<j}\gamma_{ij}Q_{ij}\otimes R_{ij}+\bar{\gamma}_{ij}Q_{ij}^{\dagger}\otimes R_{ij}^{\dagger}\nonumber \\
 & \!=\!\!\! & \sum_{i<j}\!\begin{pmatrix}\!\gamma_{ij}q_{ij}^{11}R_{ij}\!+\!\bar{\gamma}_{ij}\bar{q}_{ij}^{11}R_{ij}^{\dagger} & \phantom{-(}\gamma_{ij}q_{ij}^{12}R_{ij}\!+\!\bar{\gamma}_{ij}\bar{q}_{ij}^{21}R_{ij}^{\dagger}\phantom{)}\\
\!\gamma_{ij}q_{ij}^{21}R_{ij}\!+\!\bar{\gamma}_{ij}\bar{q}_{ij}^{12}R_{ij}^{\dagger} & -(\gamma_{ij}q_{ij}^{11}R_{ij}\!+\!\bar{\gamma}_{ij}\bar{q}_{ij}^{11}R_{ij}^{\dagger})
\end{pmatrix}\nonumber \\
 & \!=\! & \begin{pmatrix}A & B\\
B^{\dagger} & -A
\end{pmatrix}.\label{eq:mastereq}
\end{eqnarray}
Now we treat each submatrix $A$ and $B$ individually. Demanding
$A=0$ yields 
\begin{equation}
\sum_{i<j}\gamma_{ij}q_{ij}^{11}R_{ij}=-\sum_{i<j}\bar{\gamma}_{ij}\bar{q}_{ij}^{11}R_{ij}^{\dagger}\:\text{,}
\end{equation}
thus $\sum_{i<j}\gamma_{ij}q_{ij}^{11}R_{ij}$ must be skew-Hermitian.
As $R_{ij}$ has zero trace, we extract the following set of equations:
\begin{eqnarray}
\real(\sum_{i<j}\gamma_{ij}q_{ij}^{11}r_{ij}^{11}) & = & 0\:,\label{eq:lineq1}\\
\sum_{i<j}\gamma_{ij}q_{ij}^{11}r_{ij}^{12}+\sum_{i<j}\bar{\gamma}_{ij}\bar{q}_{ij}^{11}\bar{r}_{ij}^{21} & = & 0\:.
\end{eqnarray}
On the other hand, demanding $B=0$ yields
\begin{eqnarray}
\sum_{i<j}\gamma_{ij}q_{ij}^{12}r_{ij}^{11}+\sum_{i<j}\bar{\gamma}_{ij}\bar{q}_{ij}^{21}\bar{r}_{ij}^{11} & = & 0\:,\\
\sum_{i<j}\gamma_{ij}q_{ij}^{12}r_{ij}^{12}+\sum_{i<j}\bar{\gamma}_{ij}\bar{q}_{ij}^{21}\bar{r}_{ij}^{21} & = & 0\:,\\
\sum_{i<j}\gamma_{ij}q_{ij}^{12}r_{ij}^{21}+\sum_{i<j}\bar{\gamma}_{ij}\bar{q}_{ij}^{21}\bar{r}_{ij}^{12} & = & 0\:.\label{eq:lineq2}
\end{eqnarray}
Treating real and imaginary part separately, these are $3+6=9$ real
equations for the six complex values $\gamma_{ij}$.

Before continuing with the proof, we have to ensure that these equations
are linearly independent. This can be checked for by expanding the
Schmidt bases $\ket{i}_{\text{AB}}$ and $\ket{i}_{\text{CD}}$ in
terms of the computational basis, i.e. 
\begin{eqnarray}
\ket{i}_{\text{AB}} & = & \sum_{a,b=0}^{1}\mu_{ab}^{i}\ket{ab},\\
\ket{i}_{\text{CD}} & = & \sum_{c,d=0}^{1}\nu_{cd}^{i}\ket{cd},
\end{eqnarray}
where the only dependence among the $\ket{i}_{\text{AB}}$ is
\begin{equation}
\braket{i|j}_{\text{AB}}=\sum_{a,b}\mu_{ab}^{i}\bar{\mu}_{ab}^{j}=\delta_{ij}
\end{equation}
and similarly for $\nu$. Expressing the numbers $q_{ij}$ in terms
of the coefficients $\mu$,
\begin{equation}
q_{ij}^{bb^{\prime}}=\sum_{a}\mu_{ab}^{i}\bar{\mu}_{ab^{\prime}}^{j}\:,
\end{equation}
shows that the only dependence among the $q_{ij}$ is $q_{ij}^{11}=-q_{ij}^{22}$,
which has already been taken into account. Thus, the numbers $q_{ij}^{11}$,
$q_{ij}^{12}$ and $q_{ij}^{21}$ do not fulfill any additional constraints.
The same is true for the $r_{ij}$. As the orthonormal bases have
been chosen independently and randomly, the $q_{ij}$ and $r_{ij}$
are also independent from each other.

Returning to the proof, there is a three dimensional (real) solution
space for the $\gamma_{ij}$ due to Eqs.~(\ref{eq:lineq1}) to (\ref{eq:lineq2})
if we do not impose the constraints (\ref{eq:gamma}) yet. As $\gamma_{ij}=0$
for all~$i,j$ is certainly a solution, we can parametrize this solution
space by 
\begin{equation}
\gamma_{ij}=\sum_{a=1}^{3}x_{a}v_{ij}^{a}\:,\label{eq:solution3}
\end{equation}
where the $x_{a}$ are the three real-valued parameters. 

Luckily, we have additional constraints at hand as the $\gamma_{ij}$
are not independent. Let us define the normalized variables $c_{ij}:=(\lambda_{i}\lambda_{j})^{-\nicefrac{1}{2}}\gamma_{ij}$.
Then
\begin{eqnarray}
c_{ij}c_{jk} & = & (1-e^{i(\varphi_{i}-\varphi_{j})})(1-e^{i(\varphi_{j}-\varphi_{k})})\nonumber \\
 & = & 1-e^{i(\varphi_{i}-\varphi_{j})}-e^{i(\varphi_{j}-\varphi_{k})}+e^{i(\varphi_{i}-\varphi_{k})}\nonumber \\
 & = & c_{ij}+c_{jk}-c_{ik}\:,\label{eq:compatibility}
\end{eqnarray}
for all $i,j,k$. This implies also (setting $i=k$) 
\begin{equation}
\vert c_{ij}\vert^{2}=c_{ij}+\bar{c}_{ij}\:.\label{eq:cij2}
\end{equation}
Substituting for $c_{ij}$ the solution (\ref{eq:solution3}) yields
for all $i<j$
\begin{equation}
\sum_{a,b=1}^{3}x_{a}x_{b}v_{ij}^{a}\bar{v}_{ij}^{b}=\sqrt{\lambda_{i}\lambda_{j}}\sum_{a=1}^{3}x_{a}(v_{ij}^{a}+\bar{v}_{ij}^{a}).
\end{equation}
There are six equations for the three variables $x_{a}$. Taking the
four equations for $i=1$, $j=1,\ldots,4$, yields four independent
equations as each equation makes use of a different, independent Schmidt
coefficient $\lambda_{i}$. Additionally, any of the equations can
be solved for any of the $x_{a}$ and the Schmidt coefficients $\lambda_{i}$
have not been used to obtain the solutions in Eq.~(\ref{eq:solution3}).
Therefore, only the trivial solution $x_{a}=0$ exists, thus 
\begin{equation}
c_{ij}=\gamma_{ij}=0\:.
\end{equation}
Consequently, all phases $\varphi_{i}=\varphi$ must be equal. Thus
$\ket\phi=e^{i\varphi}\ket\psi$ which corresponds to the same physical
state.
\end{proof}
The same result is also true for other configurations of known marginals
that result from relabeling the particles. 

{\it The case of higher-dimensional systems.}
The proof can seamlessly be extended to the case of qudits 
having higher internal dimension $d$.

\begin{thm}
Almost all four-qudit pure states of internal dimension $d$ are uniquely
determined among pure states by the three two-body marginals of particles
$\rho_{\text{AB}}$, $\rho_{\text{CD}}$ and $\rho_{\text{BD}}$.
In particular, this implies that they are $2$-UDP.\label{thm:thm2}
\end{thm}

\begin{proof}
The proof follows exactly the same steps as in the qubit case. The
bases of the subspaces $A,B$ and $C,D$ are then $d^{2}$-dimensional,
thus $i$ and $j$ run from $1$ to $d^{2}$ and there are $d^{2}$
free phases {[}$(d^{2}-1)$ if ignoring a global phase{]}. There are
then $\frac{d^{2}(d^{2}-1)}{2}$ different complex-valued $\gamma_{ij}$
with $i<j$. The Eq.~(\ref{eq:mastereq}) consists in this case of
$d\times d$ submatrices:{\small{}
\begin{equation}
\sum_{i<j}\!\begin{pmatrix}\gamma_{ij}q_{ij}^{11}R_{ij}+\bar{\gamma}_{ij}\bar{q}_{ij}^{11}R_{ij}^{\dagger} & \ldots & \gamma_{ij}q_{ij}^{1d}R_{ij}+\bar{\gamma}_{ij}\bar{q}_{ij}^{d1}R_{ij}\\
\vdots & \ddots & \vdots\\
\gamma_{ij}q_{ij}^{d1}R_{ij}+\bar{\gamma}_{ij}\bar{q}_{ij}^{1d}R_{ij}^{\dagger} & \ldots
\end{pmatrix}=0\:\text{.}
\end{equation}
}Again, the lower left submatrices are the adjoints of the upper right
ones, thus it suffices to set the upper right ones to zero. All submatrices
on the diagonal must be skew-Hermitian, and the last diagonal matrix
can be expressed by the other diagonal entries due to tracelessness:
\begin{itemize}
\item Every off-diagonal submatrix such as $\gamma_{ij}q_{ij}^{12}R_{ij}+\bar{\gamma}_{ij}\bar{q}_{ij}^{21}R_{ij}^{\dagger}$
yields $2(d^{2}-1)$ real equations, as $R_{ij}$ is a traceless $d\times d$
matrix, thus $r_{ij}^{dd}=-r_{ij}^{11}-\ldots-r_{ij}^{d-1,d-1}$.
There are $\frac{d(d-1)}{2}$ off-diagonal submatrices on the upper
right, thus they yield $(d^{2}-1)d(d-1)$ real equations.
\item Every diagonal submatrix is skew-Hermitian, which exhibits $d+2\frac{d(d-1)}{2}=d^{2}$
real equations, and traceless, which removes one of the diagonal equations,
leaving $d^{2}-1$ equations. There are $d-1$ diagonal submatrices,
yielding a total of $(d-1)(d^{2}-1)$ real equations.
\end{itemize}
Thus, there is a total of $(d-1)(d^{2}-1)+d(d-1)(d^{2}-1)=(d^{2}-1)^{2}$
(real) equations. Consequently, the $\frac{d^{2}(d^{2}-1)}{2}$ complex-valued
$\gamma_{ij}$ are reduced to $2\frac{d^{2}(d^{2}-1)}{2}-(d^{2}-1)^{2}=d^{2}-1$
real parameters, which matches again the number of free phases in
the ansatz.

From the compatibility equations (\ref{eq:compatibility}), we can
choose those with $i=1$, $j=1\ldots d^{2}$ to obtain a set of $d^{2}$
independent quadratic equations, as there are by assumption $d^{2}$
independent Schmidt coefficients. Therefore, the only solution is
$\gamma_{ij}=0$ as in the qubit case, implying that $\ket\phi=e^{i\varphi}\ket\psi$.
\end{proof}

{\it States of $n$ particles.}
Even though above theorem is limited to states of four particles, the result sheds some light on
states of more parties.

\begin{corollary}
For $n \geq 4$, almost all $n$-qudit pure states of parties $A,B,C,D,E_1,\ldots E_{n-4}$ of internal
dimension $d$ are uniquely determined among pure states by the three $(n-2)$-body marginals of particles
$\rho_{\text{ABE}_1\ldots}$, $\rho_{\text{CDE}_1\ldots}$ and $\rho_{\text{BDE}_1\ldots}$.
In particular, this implies that they are $(n-2)$-UDP.\label{thm:cor1}
\end{corollary}

\begin{proof}
We denote by $E$ all the parties $E_1,\ldots,E_{n-4}$. Consider a generic pure
$n$-particle state $\ket\psi$ with known $(n-2)$-body marginals $\rho_{\text{ABE}}$, $\rho_{\text{ACE}}$
and $\rho_{\text{CDE}}$. From these, one can obtain the $(n-4)$-particle marginal $\rho_\text{E}$.
This allows us to decompose a generic state into
\begin{equation}
\ket\psi = \sum_{i=1}^{\min(d^4,d^{n-4})} \sqrt{\lambda_i} \ket{\psi_i} \otimes \ket{i}_\text{E},
\end{equation}
where the Schmidt basis $\ket{i}_\text{E}$ and Schmidt coefficients $\lambda_i$ are determined by
$\rho_\text{E}$ and the Schmidt vectors $\ket{\psi_i}$ on $ABCD$ have yet to be determined.
On the one hand, knowing the $(n-2)$-body marginal $\rho_{\text{ABE}}$ allows us to determine 
all expectation values of the form
\begin{equation}
\braket{\psi|O_\text{A} \otimes O_\text{B} \otimes \ketbra{i}{i}_\text{E}|\psi} = \trace(O_\text{A} \otimes O_\text{B} \otimes \ketbra{i}{i}_\text{E}\rho_\text{ABE})
\end{equation}
for all $i$, where $O_\text{A}$ and $O_\text{B}$ are some local observables of parties $A$ and $B$, respectively.
On the other hand, this is equivalent to knowing all expectation values $\braket{\psi_i|O_\text{A} \otimes O_\text{B}|\psi_i}$
of the pure four-particle constituent $\ket{\psi_i}$, yielding its reduced state $\rho_\text{AB}^{(i)}$. The same can be
done for parties $AC$ and parties $CD$. According to Theorem \ref{thm:thm2}, this determines the states
$\ket{\psi_i}$ uniquely up to a phase. Thus, the joint state on $ABCDE$ has to have the form
\begin{equation}
\ket{\psi^\prime} = \sum_{i=1}^{\min(d^4,d^{n-4})}e^{i\varphi_i}\sqrt{\lambda_i} \ket{\psi_i} \otimes \ket{i}_\text{E}.
\end{equation}
However, from this family only the choice $\varphi_i=\varphi_j$ for all $i,j$ is compatible with the known
reduced state $\rho_{\text{ABE}}$: The reduced state
\begin{equation}
\rho_{\text{ABE}}^\prime = \sum_{i,j} e^{i(\varphi_i-\varphi_j)}\sqrt{\lambda_i\lambda_j} \trace_{\text{CD}}(\ketbra{\psi_i}{\psi_j}) \otimes \ketbra{i}{j}_\text{E}
\end{equation}
can be compared term by term with the known marginal, as the $\ket{i}_\text{E}$ are orthogonal. 
Therefore, $\ket{\psi^\prime}=e^{i\varphi}\ket\psi$ and the state is determined again.
\end{proof}

It must be stressed that the main statement of this Corollary is the fact that 
three $n-2$ marginals can already suffice. The fact that pure states are $(n-2)$-UDP 
is not surprising, as usually already less knowlegde is sufficient to make a pure 
state UDA, see Ref.~\cite{jones2005parts} for a discussion.

\begin{figure}[t!]
\begin{centering}
\includegraphics[width=0.85\columnwidth]{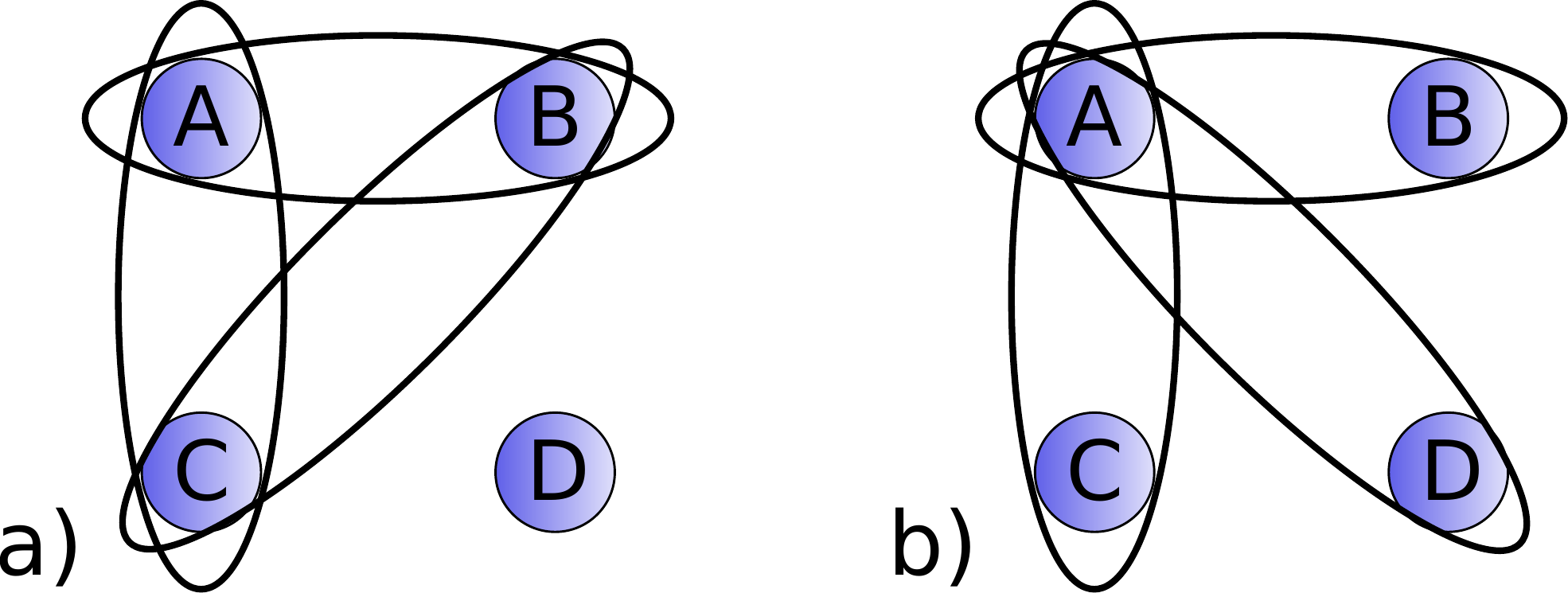}
\par\end{centering}
\protect\caption{Illustration of the two other possible sets 
of three two-body marginals: a) a set of marginals, which clearly 
does not determine the global state, as $\rho_D$ is not fixed. 
b) a set of marginals to which our proof does not apply. Nevertheless, 
we have numerical evidence that these marginals still determine the 
state uniquely for qubits.
\label{fig:Illustration-unknown}}
\end{figure}

{\it States that are not UDP.}
As the proof above is valid for generic states only, it is natural to
ask whether there are special four-particle states that are not UDP. This
is indeed the case. In the following, we give an incomplete list of
undetermined four-particle qubit states. Note that if any two states $\ket\psi$
and $\ket\phi$ share the same two-body marginals, then also all local
unitary equivalent states $\ket\psi^{\prime}=U_{A}\otimes U_{B}\otimes U_{C}\otimes U_{D}\ket\psi$
and $\ket\phi^{\prime}=U_{A}\otimes U_{B}\otimes U_{C}\otimes U_{D}\ket\phi$
share the same marginals. Thus, we restrict ourselves to states $\ket\psi=\sum\alpha_{ijkl}\ket{ijkl}$
of the standard form introduced in Ref.~\citep{carteret2000multipartite},
where 
\begin{eqnarray}
\alpha_{0000},\alpha_{0001},\alpha_{0010},\alpha_{0100},\alpha_{1000} & \in & \mathbb{R}\:,\nonumber \\
\alpha_{0111},\alpha_{1011},\alpha_{1101},\alpha_{1110} & = & 0
\end{eqnarray}
and all other coefficients being complex. In the following list, the
states are always assumed to be normalized. To shorten the notation,
we make use of the $W$ state 
\[
\ket{W_{4}}=\frac{1}{2}(\ket{0001}+\ket{0010}+\ket{0100}+\ket{1000})
\]
and of the Dicke state 
\begin{eqnarray*}
\ket{D_2^4} & = & \frac{1}{\sqrt{6}}(\ket{0011}+\ket{0101}+\ket{1001}\\
 &  & \hphantom{\frac{1}{\sqrt{6}}(}+\ket{0110}+\ket{1010}+\ket{1100}).
\end{eqnarray*}
Due to the standard form, we have in the following $a,b\in\mathbb{R}$, 
while $r,s\in\mathbb{C}$. The claimed properties of the states can directly
be computed. 
\begin{itemize}
\item For fixed $a,b$ and $s$, the family 
\begin{equation}
\ket\psi=a\ket{0000}+b\ket{W_{4}}+se^{i\varphi}\ket{1111}
\end{equation}
shares the same two-body marginals for all values of $\varphi$.
\item For the same state with $a=0,b=\frac{2}{\sqrt{6}}$ and $s=\frac{1}{\sqrt{3}}$,
\begin{equation}
\ket\phi=\frac{1}{2}\ket{0000}+\frac{1}{\sqrt{2}}e^{i\varphi}\ket{D_2^4}-\frac{1}{2}e^{2i\varphi}\ket{1111}
\end{equation}
shares the same marginals for all values of $\varphi$.
\item For every state 
\begin{equation}
\ket\psi=a\ket{0000}+r\ket{D_2^4}+s\ket{1111},
\end{equation}
the state 
\begin{equation}
\ket\phi=a\ket{0000}+re^{i\varphi_{r}}\ket{D_2^4}+se^{i\varphi_{s}}\ket{1111}
\end{equation}
 shares the same marginals if $\overline{r}se^{i\varphi_{s}}=are^{i\varphi_{r}}(1-e^{i\varphi_{r}})
 +\overline{r}se^{i\varphi_{r}}$,
which is feasible for e.g., $a=0$.
\end{itemize}
All of our examples are superpositions of Dicke states and generalized GHZ states. By a local unitary operation,
these examples also include the Dicke state with three excitations.
The examples prove that Theorem 3 does not hold for all four-particle states, but only for 
generic states. 

{\it Discussion.}
We have shown that generic four-qudit pure states are uniquely determined
among pure states by three of their six different marginals of two
parties. Interestingly, from this it follows that pure states of an
arbitrary number of qudits are determined by certain subsets of their
marginals having size $n-2$.
The proof required two marginals of distinct systems to be
equal, for instance $\rho_{\text{AB}}$ and $\rho_{\text{CD}}$, in
order to fix the Schmidt decomposition of the compatible state. However,
there are two other sets of three two-body marginals, illustrated
in Fig.~\ref{fig:Illustration-unknown}. The first one, namely knowledge
of $\rho_{\text{AB}}$, $\rho_{\text{AC}}$ and $\rho_{\text{BC}}$, is certainly
not sufficient to fix the state, as we do not have any knowledge of
particle $D$ in this case: Every product state $\rho_{\text{ABC}}\otimes\rho_{\text{D}}$
with arbitrary state $\rho_{\text{D}}$ is compatible. The situation
for the second configuration, namely knowledge of the three marginals
$\rho_{\text{AB}}$, $\rho_{\text{AC}}$ and $\rho_{\text{AD}}$,
is not that clear. In a numerical survey testing random four-qubit 
states, we could not find pairs of different pure states which coincide
on these marginals. Thus, we conjecture that any marginal 
configuration involving all four parties determines generic states.
In any case, knowledge of any set of four two-body marginals fixes
the state, as there are always two marginals of distinct
particle pairs present in these sets.

The question remains which pure four-qubit states are also 
uniquely determined among all mixed states by their two-body marginals. The 
results from Ref.~\citep{jones2005parts} suggest that generic states 
are not UDA, and Ref.~\citep{huber2016characterizing} 
shows that for the case of four qutrits
and knowledge of all marginals, as well as for four qubits
and the special marginal configuration of Fig. 1 (b), generic
states are not UDA.
On the other hand, in the same reference, a numerical procedure indicated that
for generic pure four-qubit states the compatible mixed states (having the same marginals)
are never of full rank. Clarifying this situation is an interesting problem for further 
research.

{\it Acknowledgments.}
This work was supported by 
the Swiss National Science Foundation (Doc.Mobility Grant 165024), 
the COST Action MP1209, the DFG, and the ERC (Consolidator Grant
No. 683107/TempoQ).

\bibliographystyle{aipnum4-1}
\bibliography{cite}

\end{document}